\documentclass[12pt]{article}

\usepackage{scicite}
\usepackage{amsmath}
\usepackage{amsfonts}
\usepackage{amssymb}
\usepackage{graphicx}

\usepackage{times}

\newtheorem{theorem}{Theorem}

\newenvironment{proof}[1][Proof]{\noindent\textbf{#1.} }{\ \rule{0.5em}{0.5em}}

\newenvironment{sciabstract}{%
\begin{quote} \bf}
{\end{quote}}

\title{A Dynamical Systems Approach to Cryptocurrency Stability}

\author
{Carey Caginalp,$^{1, 2\ast}$\\
\\
\normalsize{$^{1}$Department of Mathematics, University of Pittsburgh, USA}\\
\normalsize{301 Thackeray Hall, Pittsburgh PA 15260, USA}\\
\normalsize{$^{2}$Economic Science Institute, Chapman University}\\
\normalsize{One University Drive, Orange CA 92866, USA}\\
\\
\normalsize{$^\ast$To whom correspondence should be addressed; E-mail:  carey\_caginalp@alumni.brown.edu.}
}

\date{}

\begin{document} 



\maketitle


\begin{sciabstract}
Recently, the notion of cryptocurrencies has come to the fore of 
public interest. These assets that exist only in electronic form, 
with no underlying value, offer the owners some protection 
from tracking or seizure by government or creditors. We model 
these assets from the perspective of asset flow equations developed 
by Caginalp and Balenovich, and investigate their stability under 
various parameters, as classical finance methodology is 
inapplicable. By utilizing the concept of liquidity price 
and analyzing stability of the resulting system of ordinary 
differential equations, we obtain conditions under which the system 
is linearly stable. We find that trend-based motivations and 
additional liquidity arising from an uptrend are destabilizing forces, 
while anchoring through value assumed to be fairly recent price 
history tends to be stabilizing.
\end{sciabstract}


\section*{Introduction}

Blockchain technology enables large numbers of participants to make electronic
transactions directly without intermediaries, and has led, in recent years to
a new form of payment, and essentially to a new set of currencies called
cryptocurrencies. During 2017 the spectacular nine-fold rise in the price of
Bitcoin focused the spotlight of public attention on cryptocurrencies that
evolved into a new asset class. Following the pattern of other nascent assets,
speculators dominated trading and pushed prices toward a bubble.

As with some other asset bubbles of the past, notably the dot-com frenzy of
the late 1990s, the emergence of a new technology clouded judgements about the
basic value of the asset.

In almost all cases, unlike traditional equities, cryptocurrencies have no
tangible value, and even their creation is often shrouded in mystery
\cite{BER17}. For Bitcoin, the creation of additional units is relegated to a
process termed mining, named after the old technology of mining gold. In this
electronic version, computing power is used to solve complex mathematical
problems, and once solutions are found, the miner is rewarded with some units
of the cryptocurrency. It is peculiar that this brand new technology is
coupled with the organizational hierarchy of centuries ago that featured the
hegemony of gold miners and traders who held power over the economic lives of
people without any accountability. Thus, the advances in technology are
coupled with a regressive organizational structure. Major currencies such as
the US dollar are controlled by officials appointed by elected
representatives; assets such as common stocks are governed by a board of
directors that are elected by the shareholders whose rights are assured by
law, albeit through a circuitous process. Changes made by cryptocurrencies are
often made at an ad hoc meeting of the developers or miners, reminiscent of
tribal chiefs meetings of more primitive eras.

In certain types of securities, i.e. exchange-traded funds (ETFs), a certain
group of traders known as authorized participants have the capacity to demand
the underlying shares, preventing the price of the ETF from straying too far
from its fundamental value. Securities held in a brokerage account are insured
up to a value of several million dollars by the federal government against
circumstances such as hacking or company insolvency.

Conversely, for cryptocurrencies, many of these rules have been lacking and a
few have been recently instituted. First, these cryptocurrencies have no
underlying assets. Further, there exists no mechanism by which one can redeem
any value from a bank or other institution. Second, even the country of
residence (let alone a business address) of the creators is merely
speculative, making it unclear to which court one could possibly appeal in the
case of a grievance. The original "developer" for Bitcoin, for example, is
known only by a pseudonym. Third, the individual investor has no influence
over something as simple as the number of Bitcoins in existence. Instead,
these decisions are generally relegated to the heads of electronic Bitcoin
mining operations, whose interests may be disjoint from investors'. Bitcoin is
currently set to be capped at 21 million units, but there is no legal obstacle
to prevent an increase at the whim of the miners. Fourth, many individual
investors and some academicians tacitly assume that the laws and protections
afforded by the purchase of securities through stock exchanges such as NYSE
must also apply to cryptocurrencies. Recently, the Securities and Exchange
Commission (SEC) has shuttered some initial coin offerings (ICOs)
\cite{BIG17,SHI17} and announced future regulations to attempt to inject more
clarity into the marketplace. Finally, Bitcoin is vulnerable to hacking or
simply forgetfulness. Many holders of Bitcoin have their information stored in
exchanges, i.e. platforms that handle transaction and storage of the
cryptocurrency. These are hijacked with a disturbing regularity, and
litigation can be pending years later, as in the infamous Mt. Gox hacking of
2014 \cite{MCM14}.

Cryptocurrencies offer both opportunities and risks to society. On the one
hand, cryptocurrencies and technology underpinning them -- if designed
appropriately -- could be used to make transactions faster, safer and cheaper,
alongside other societal benefits \cite{COH18,HUT17}. \ A less apparent
feature is that they can make it more difficult (though not theoretically
impossible; see \cite{BOH16,ENS16}) for totalitarian governments to
expropriate savings, either directly or indirectly through currency inflation,
thereby depriving savers of a large fraction of their assets. In this way, a
proper cryptocurrency could lead to greater economic freedom, and render more
difficult the financing of a dictatorship. Indeed, this can be modelled by a
choice of two alternatives: either their home currency or cryptocurrency that
cannot easily be seized \cite{GET96,STO06}.

The risks presented by existing cryptocurrencies are multi-faceted. The
difficulty in tracing transactions facilitate illicit activity and its
financing. A less obvious -- and possibly the most significant -- risk arises
from the instability of prices of major cryptocurrencies. As the market
capitalization (number of units times the price of each unit) of the
cryptocurrencies rises, there is growing risk that a sharp drop in the price
of a cryptocurrency could have a cascading effect on other sectors of world
economy, particularly if borrowing is involved. During the period October 2017
to April 2018, the price of a Bitcoin unit rose from \$6,000 to \$20,000 and
back to \$6,000. The market capitalization of all cryptocurrencies during that
time period increased from \$170 billion to \$330 billion, peaking together
with Bitcoin in December 2017. While attention is often focused on the rise
and fall of the trading prices of these assets, the magnitude of the problem
of stability has increased significantly during this six month period. As
people become more accustomed to using these instruments, the market
capitalization may increase to several trillion -- i.e., a few percent of the
\$\ 75 trillion Gross World Product -- and many of the challenges will be critical.

Generally, the features of a financial instrument that might make it
attractive to speculators are undesirable to those who seek to use it as a
currency in daily transactions. Speculators see a greater opportunity in a
volatile market, as they can use technical analysis and expertise to profit at
the expense of the layperson. Conversely, large fluctuations on a day-to-day
basis create obstacles for common purchases or the pricing of service
contracts \cite{STO05}. Without stability in the marketplace, the
cryptocurrencies may simply become "a mechanism for a transfer of wealth from
the late-comers to the early entrants and nimble traders" \cite{CAG18a}. Thus,
a set of questions of critical importance deals with the potential stability
(or lack thereof) of Bitcoin or other cryptocurrencies, which is the main
topic of our paper.

The turbulence arising from the collapse of the housing bubble was a major
challenge for markets, but from a scientific perspective, it could be
addressed largely with classical methods \cite{FAM70,SCH00,SHI81}. However,
classical methods are not readily adaptable to studying cryptocurrencies, as
discussed below. We use a modern approach whereby an equilibrium price can be
determined and the stability properties established within a dynamical system
setting
\cite{CAG01,CAG98,CAG99,EHR12,JOS08,KAM89,POR94,PRI93,SHI81,SMI88,STO03,WAT81}.

\section{Modelling prices and stability}

Most of classical finance such as the Black-Scholes option pricing model has
its origin in the basic equation
\begin{equation}
\frac{1}{P}dP=\mu dt+\sigma dW\label{BS}%
\end{equation}
for the change in the relative price $P^{-1}dP$ in terms of the expected
return, $\mu,$ the standard deviation of the return, $\sigma,$ and independent
increments of Brownian motion, $dW.$ It is widely acknowledged that this
equation does not arise from compelling microeconomic considerations, nor
empirical data. But rather, it is mathematically convenient and elegant for
expressing and proving theorems (see \cite{CHA13} for discussion). Much of
risk assessment is based upon this model with an increasing array of adjustments.

The limitations of this basic model are apparent, for example, if one examines
the standard deviation of daily relative changes in the S\&P 500 index, which
is typically around $0.75\%.$ This leads to the conclusion that a $4.5\%$ drop
is a sixth standard deviation event, i.e., it occurs once every billion
trading days, while empirical data shows it is on the order of a few times per
thousand \cite{BAN18}.

Thus identifying risk on a large time scale based on the variance of a small
time scale can vastly underestimates risk.

Furthermore, the modeling of asset prices is generally based on the underlying
assumption of infinite arbitrage. While there may be some investors who are
prone to cognitive errors or bias in assessing value, the impact of their
trades will be marginalized by more savvy investors who manage a large pool of
money. Of course the inherent assumption is that there is some value to an
asset, based for example on the projection of the dividend stream, replacement
value, etc., and that the shareholder has a vote that allows him ultimately to
extract this value. For assets such as US\ Treasury bills, the model works
quite well, as the owner is assured of receiving a particular dollar amount
from the US at a specified time.

Herein lies the central problem for the application of classical theory to
cryptocurrencies:\ there is no underlying asset value, as noted above.
Cryptocurrencies constitute the opposite end of the market spectrum to
US\ Treasury Bills, in which an arbitrageur can confidently buy or sell short
based on a clear contract that will deliver a fixed amount of cash at a
predetermined time.

If fact, classical game theory would conclude that since everyone knows the
structure of the cryptocurrency, and understands that everyone else is also
aware, then the price should never deviate much from zero. Furthermore,
classical finance expressed through (\ref{BS}) would suggest that there is
some measure of risk based on the historical average of $\sigma,$ which will
be less helpful that it is for stock indexes as discussed above.

Our analysis begins with the fact that despite the absence of underlying
assets or backing, various groups have incentive to use it over traditional
currencies. In particular there are large groups who need to make transactions
outside of the usual banking system. Among these are (i) people with poor
credit who cannot obtain a credit or even a debit card, (ii) citizens of
totalitarian countries who fear expropriation of their savings, (iii) citizens
of countries with high inflation and a much lower interest rate, (iv) people
engaged in illicit activity, (v) people who espouse utilizing a new idea or technology.

Collectively, these groups constitute a core ownership of cryptocurrencies,
investing a sum that gradually grows with familiarity \cite{CLI17,DIN17,JUN17}%
. Meanwhile, the rising prices catch the attention of speculators who provide
additional cash into the system, but also bring motivations inherent in
speculation, namely momentum trading, or the tendency to buy as prices rise,
and analogously sell as prices fall \cite{TIR82}.

We assume a single cryptocurrency and that the price is established by supply
and demand without infinite arbitrage, and apply a modern theory of asset flow
\cite{CAG99}. This alternative approach relies on the notion of liquidity
price. The experimental asset markets presented a puzzle to the economics
community by demonstrating the endogenous price bubbles in which prices soared
well above any possible expectation of outcome \cite{POR94}. Caginalp and
Balenovich \cite{CAG99} observed that in addition to the trading price and
fundamental value (defined clearly by the experimental setup), there was an
additional important quantity with the same units: the total cash in the
system divided by the number of shares. Denoting this by liquidity value or
price, $L,$ they adapted earlier versions \cite{CAG90} of the asset flow model.

This approach leads to a system of ordinary differential equations, as
summarized below, whereupon equilibrium points can be evaluated and their
stability established as a function of the basic parameters.

\section{Modeling Cryptocurrency with Asset Flow Equations}

For brevity, we first present the full model which will be a nonlinear
evolutionary system that is based on \cite{CAG99} but with some key
differences for cryptocurrencies. We can then consider simpler models in which
some features are marginalized by setting parameters to zero and obtaining
$2\times2$ or $3\times3$ systems, enabling us to understand the key factors in stability.

We denote the trading price by $P\left(  t\right)  $, the number of units by
$N\left(  t\right)  $, the amount of cash available by $M\left(  t\right)  $,
and the liquidity price by $L\left(  t\right)  =M\left(  t\right)  /N\left(
t\right)  $. With $B$ as the fraction of wealth in the cryptocurrency, i.e.,
$B=NP/\left(  NP+M\right)  $, the supply and demand are given by $S=\left(
1-k\right)  B,$ $D=k\left(  1-B\right)  $ respectively, where $k$ is the
transition rate from cash to the asset. Using a standard price equation
\cite{WAT81} we write%
\begin{equation}
\tau_{0}\frac{1}{P}\frac{dP}{dt}=\frac{D}{S}-1.
\end{equation}
It follows that $B\left(  1-B\right)  ^{-1}=NP/M=P/L$ , so that the price
equation is%
\begin{equation}
\tau_{0}\frac{1}{P}\frac{dP}{dt}=\frac{k}{1-k}\frac{L}{P}-1\ .
\end{equation}
The variable $k$ is assumed to be a linearization of a tanh type function and
involves the motivations of the traders which are expressed through sentiment,
$\zeta=\zeta_{1}+\zeta_{2}$ where $\zeta_{1}$ is the trend component and
$\zeta_{2}$ is the value component. This construct has been studied, for
example, in closed-end funds \cite{AND02,CHE93,CHO93,LEE91} which frequently
trade either at a discount or premium to their net asset value. Writing the
term $k/\left(  1-k\right)  $ in terms of the $\zeta_{1}$ and $\zeta_{2}$ and
linearizing we have then%
\begin{equation}
\frac{k}{1-k}\tilde{=}1+2\zeta_{1}+2\zeta_{2}%
\end{equation}
and the price equation is then
\begin{equation}
\tau_{0}\frac{1}{P}\frac{dP}{dt}=\left(  1+2\zeta_{1}+2\zeta_{2}\right)
\frac{L}{P}-1\ .\ \label{priceEqn}%
\end{equation}

One defines $\zeta_{1}$ through two parameters, $c_{1},$ that expressed the
time scale of the trend following and $q_{1}$ the amplitude of this factor, as%
\begin{equation}
\zeta_{1}\left(  t\right)  =\frac{q_{1}}{c_{1}}\int_{-\infty}^{t}e^{-\left(
t-\tau\right)  /c_{1}}\frac{1}{P\left(  \tau\right)  }\tau_{0}\frac{dP\left(
\tau\right)  }{d\tau}d\tau\label{zeta1def}%
\end{equation}
Note that $L$ and $\zeta_{1}$ are linear functions of one another, but we
retain $L$ as a variable so the system is more easily generalized to
incorporate a time-dependent $L_{0}$. The valuation is more subtle for a
cryptocurrency. The only concept of value relates to fairly recent trading
prices. The first purchase with Bitcoin was for two slices of pizza for 10,000
Bitcoins \cite{PRI17}. The sense of value at that time was probably much less
than 2018 when people became accustomed to prices in the thousands of dollars.
We thus stipulate the definitions%

\begin{equation}
P_{a}\left(  t\right)  =\frac{1}{c_{3}}\int_{-\infty}^{t}e^{-\left(
t-\tau\right)  /c_{3}}P\left(  \tau\right)  d\tau,
\end{equation}%
\begin{equation}
\zeta_{2}\left(  t\right)  =\frac{q_{2}}{c_{2}}\int_{-\infty}^{t}e^{-\left(
t-\tau\right)  /c_{2}}\frac{P_{a}\left(  \tau\right)  -P\left(  \tau\right)
}{P\left(  \tau\right)  }d\tau,
\end{equation}
i.e., $\zeta_{2}$ represents the motivation to buy based on the discount from
the perceived value of the cryptocurrency. Finally, the liquidity will not be
constant but will be the sum of the core group's capital $L_{0}$ plus the
additional amounts arriving from speculators that is correlated with the
recent trend:%
\begin{equation}
L\left(  t\right)  =L_{0}+\frac{L_{0}}{c}q\int_{-\infty}^{t}e^{-\left(
t-\tau\right)  /c}\frac{\tau_{0}}{P\left(  \tau\right)  }\frac{dP\left(
\tau\right)  }{d\tau}d\tau\label{ldef}%
\end{equation}
We assume that $L_{0}$ is constant, but one can easily adapt the model to
include temporal changes in $L_{0}$ due to, for example, greater public
acceptance of cryptocurrencies. By differentiating (\ref{zeta1def}-\ref{ldef})
and combining the resulting equations with (\ref{priceEqn}) we obtain the 5x5
system of ordinary differential equations:%
\begin{align}
c_{3}P_{a}^{\prime} &  =P-P_{a},\nonumber\\
c_{2}\zeta_{2}^{\prime} &  =q_{2}\frac{P_{a}\left(  t\right)  -P\left(
t\right)  }{P_{a}\left(  t\right)  }-\zeta_{2},\nonumber\\
\tau_{0}P^{\prime} &  =\left(  1+2\zeta_{1}+2\zeta_{2}\right)  L-P,\nonumber\\
cL^{\prime} &  =1-L+q\left\{  \left(  1+2\zeta_{1}+2\zeta_{2}\right)
L-P\right\}  ,\nonumber\\
c_{1}\zeta_{1}^{\prime} &  =q_{1}\left(  \left(  1+2\zeta_{1}+2\zeta
_{2}\right)  \frac{L}{P}-1\right)  -\zeta_{1}.\label{fullsystemnon}%
\end{align}
We find a unique equilibrium at $\left(  P,P_{a},L,\zeta_{1},\zeta_{2}\right)
=\left(  1,1,1,0,0\right)  $. In other words, the only steady-state of the
system occurs when the price, the anchoring notion of fundamental value, and
liquidity price all coincide with the base liquidity value $L_{0}$
\cite{SHE05}. The time scale for price adjustment will be short as markets
adjust rapidly to supply/demand changes. Much longer will be the time scale
for observing the trend and reacting to under or over-valuation, and the
assessment of valuation anchored through weighted price averages. Moreover,
one might expect that the valuation is on an even longer time scale. Thus one
expects three time scales such that $\tau_{0}\ll c,c_{1},c_{2}\ll c_{3}$,
which we can scale as $c=c_{1}=c_{2}=1,$ and we allow arbitrary $\tau
_{0},c_{3}$ in the analysis.%
\begin{equation}
\left(
\begin{array}
[c]{c}%
\tau_{0}P\\
c_{3}P_{a}\\
L\\
\zeta_{1}\\
\zeta_{2}%
\end{array}
\right)  =\left(
\begin{array}
[c]{ccccc}%
-1 & 0 & 1 & 2 & 2\\
1 & -1 & 0 & 0 & 0\\
-q & 0 & q-1 & 2q & 2q\\
-q_{1} & 0 & q_{1} & 2q_{1}-1 & 2q_{1}\\
-q_{2} & q_{2} & 0 & 0 & -1
\end{array}
\right)  \left(
\begin{array}
[c]{c}%
P\\
P_{a}\\
L\\
\zeta_{1}\\
\zeta_{2}%
\end{array}
\right)  .\label{fullsystem}%
\end{equation}
Thus, the system is determined entirely by three parameters: $q$, the
attention to trend; $q_{1}$, a measure of the influence of delay times; and
$q_{2}$, the influence of fundamental value, along with the time parameters
$\tau_{0}$ and $c_{3}$. The question of stability can be investigated by
calculating the eigenvalues in the relevant parameter space, i.e. $\left(
q,q_{1},q_{2}\right)  \in\mathbb{R}_{+}^{3}$ (the first octant), along with
$\tau_{0}$ and $c_{3}$. In particular, the main question is whether the
maximal real part of the eigenvalues is positive, leading to instability, or
if they are all negative, yielding stability. One sees that there is a double
eigenvalue at $\lambda=-1$, and the other three eigenvalues remain negative if
the Routh-Hurwitz conditions \cite{SUR10} below are satisfied%
\begin{align}
\frac{1}{\tau_{0}}+\frac{1}{c_{3}}+Q  & >0,\nonumber\\
\left(  \frac{Q}{c_{3}}+\frac{1}{\tau_{0}}+2\frac{q_{2}}{\tau_{0}}+\frac
{1}{\tau_{0}c_{3}}\right)  \left(  \frac{1}{\tau_{0}}+\frac{1}{c_{3}%
}+Q\right)    & >\frac{1}{c_{3}\tau_{0}}.\label{5x5rh1}%
\end{align}
where we have set $Q=1-q-2q_{1}$. A sufficient set of conditions for
(\ref{5x5rh1}) to hold is the following:%
\begin{align}
\frac{1}{c_{3}}+\frac{1}{\tau_{0}}  & >q+2q_{1}=:K,\nonumber\\
\frac{1}{c_{3}}+\frac{1}{\tau_{0}}  & >\frac{K}{c_{3}}-2\frac{q_{2}}{\tau_{0}%
}.\label{5x5rh2}%
\end{align}
However, one can observe numerically that (\ref{5x5rh2}) are not necessary
conditions to satisfy (\ref{5x5rh1}). Also, if we set $q_{2}=0,$ we obtain the
simpler condition%
\begin{equation}
\frac{1}{c_{3}}+\frac{1}{\tau_{0}}>K
\end{equation}
for stability, which we will see describes a simpler model that excludes
valuation and the component of investor sentiment associated with it. We
sketch various cross-sections holding one of these parameters constant and
numerically computing eigenvalues across values of the other two. Note in 
Figure \ref{ptb} below that increasing $q_{2}$ induces a stabilizing effect, while
large $K$ serves to make the system less stable. We choose various values of
$\tau_{0}$ and $c_{3}$ in Figure \ref{Fig 5x5}.

This yields a number of results. First, as market participants focus greater
attention to the deviation of the asset from the acquired fundamental value
driven from the liquidity price, there is less room for prices to stray from
equilibrium. In addition, for a fixed $q_{2}$, the asset would experience
stability given that $K$ is large enough. Finally, for $K$ large enough, one
sees that we have instability for a large range of $q$, i.e., if investors
place too much emphasis on the relative trend, the asset price becomes
unstable. The shaded regions indicate the range of parameters for which the
system (\ref{fullsystem}) is stable.

When we set $q_{2}=0$, the model simplifies somewhat, leaving a linear
interface between the regions of stability and instability. We then have the
following theorem. We define $Q:=1-q-2q_{1}$.

\begin{theorem}
Consider the system (\ref{fullsystem}). With $q_{2}=0$, one has stability of
the system (\ref{fullsystem}) if and only if
\begin{equation}
Q+\frac{1}{\tau_{0}}>0\label{thmq2}%
\end{equation}
.
\end{theorem}

\begin{proof}
Setting $q_{2}=0$, the necessary conditions become%
\begin{equation}
\left(  Q+\frac{1}{\tau_{0}}\right)  \left(  \frac{1}{\tau_{0}}+\frac{Q}%
{c_{3}}+\frac{1}{c_{3}\tau_{0}}+\frac{1}{c_{3}^{2}}\right)  >0\text{ and
}Q+\frac{1}{\tau_{0}}+\frac{1}{c_{3}}>0\label{condq2}%
\end{equation}
We prove this is equivalent to $Q+\frac{1}{\tau_{0}}>0$.

(i) Assume $Q+\frac{1}{\tau_{0}}>0$. Then clearly the second inequality in
(\ref{condq2}) is satisfied. Also, one has%
\begin{equation}
\frac{1}{\tau_{0}}+\frac{Q}{c_{3}}+\frac{1}{c_{3}\tau_{0}}+\frac{1}{c_{3}^{2}%
}=\left(  Q+\frac{1}{\tau_{0}}\right)  \left(  \frac{1}{c_{3}}\right)
+\frac{1}{\tau_{0}}+\frac{1}{c_{3}^{2}}>0,
\end{equation}
satisfying the first inequality.

(ii) Suppose (\ref{condq2}) holds. Then clearly%
\begin{equation}
0<\left(  Q+\frac{1}{\tau_{0}}+\frac{1}{c_{3}}\right)  \frac{1}{c_{3}}%
+\frac{1}{\tau_{0}}=\frac{1}{\tau_{0}}+\frac{Q}{c_{3}}+\frac{1}{c_{3}\tau_{0}%
}+\frac{1}{c_{3}^{2}},
\end{equation}
implying (\ref{thmq2}).
\end{proof}

\bigskip

\section{The effect of liquidity with or without sentiment}

In order to isolate the effect of liquidity, we eliminate the role of investor
sentiment and value by setting the associated parameters to zero. To this end,
we are left with the system%
\begin{align}
\tau_{0}P^{\prime} &  =L-P,\nonumber\\
cL^{\prime} &  =1+\left(  q-1\right)  L-qP.
\end{align}
One readily calculates that there will be positive eigenvalues of the
linearized system if and only if $q>1+\frac{c}{\tau_{0}}$ In other words, in a
system where only price and liquidity are relevant, a large amplitude $q$ of
liquidity is destabilizing while a large time scale for the liquidity is
stabilizing. The stability is illustrated in the Figure \ref{simplestmodel}.

Another nontrivial subcase is obtained from examining the full model
(\ref{fullsystemnon}) in the case where we set the value component of the
sentiment, $\zeta_{2}$, and the fundamental value equal to zero. We then have
the system of equations%
\begin{align}
\tau_{0}\frac{dP}{dt} &  =\left(  1+2\zeta_{1}\right)  L-P\nonumber\\
c\frac{dL}{dt} &  =1-L+q\left(  1+2\zeta_{1}\right)  L-qP\nonumber\\
c_{1}\frac{d\zeta_{1}}{dt} &  =q_{1}\left(  1+2\zeta\right)  \frac{L}{P}%
-q_{1}-\zeta_{1}\label{10}%
\end{align}

One then observes that the only equilibrium point is $L=P=L_{0}$ and $\zeta
=0$. Recalling that $Q:=1-q-2q_{1}$, one has the following.

\begin{theorem}
The system (\ref{10}) incorporating valuation and sentiment (with $c:=c_{1}$)
is stable if and only if
\begin{equation}
Q+\frac{c}{\tau_{0}}>0,
\end{equation}
i.e. if the perturbations from trend and valuation sentiment are sufficiently
small as a relative comparison to the timescale of reaction with respect to price.
\end{theorem}

\begin{proof}
By scaling, assume without loss of generality that $c_{1}=c=1$; then we can
linearize the system as follows:%
\begin{equation}
\left(
\begin{array}
[c]{c}%
P\\
L\\
\zeta
\end{array}
\right)  ^{\prime}=\left(
\begin{array}
[c]{ccc}%
-1/\tau_{0} & 1/\tau_{0} & 2/\tau_{0}\\
-q & q-1 & 2q\\
-q_{1} & q_{1} & 2q_{1}-1
\end{array}
\right)  \left(
\begin{array}
[c]{c}%
P\\
L\\
\zeta
\end{array}
\right)  =:A\left(
\begin{array}
[c]{c}%
P\\
L\\
\zeta
\end{array}
\right)  .
\end{equation}

Leaving aside the eigenvalue of $-1$ that is present for all values of the
parameters, the matrix $A$ has eigenvalues with positive real part if and only
if%
\begin{equation}
q+2q_{1}>1+\frac{1}{\tau_{0}}.
\end{equation}
After rescaling, this is the statement of the theorem.
\end{proof}

Furthermore, we have either zero or two roots with positive real parts, so
that we will have a stable spiral for $Q+\frac{c}{\tau_{0}}>0$ and an unstable
spiral for $Q-\frac{c}{\tau_{0}}<0$ for the equilibrium point at $\left(
1,1,0\right)  $. This matches our intuition from an economics perspective
since one has instability when $q+2q_{1}>1+\frac{c}{\tau_{0}}$, i.e., there
will be stability if $q+2q_{1}<1$ regardless of $c$ and $\tau_{0}$. For
$q+2q_{1}>1$, one sees that instability arises when $\frac{c}{\tau_{0}}$ is
sufficiently small, i.e. traders are focused on short term trends.

The analysis above clearly shows that the potential stability of a
crypto-asset may be contingent on several parameters that one may be able to
influence. With this information, further research may be useful to examine
the correlations and fit of these parameters with the effects of news and
government policy. A\ problem of future interest would be whether, and if so
how, governmental policy might be developed to diminish the volatility in
cryptocurrencies. Another alternative would be a decentralized cryptocurrency
with a concrete value. A good index to base this on would be either current or
future gross world product (which could be estimated via futures markets). For
a nominal fee, holders of this currency would be able to demand a basket of
underlying currencies (such as dollar, euro, yen, etc.), which would keep the
value of such a currency relatively close to its true fundamental value.

\bigskip

\begin{quote}
{\bf References and Notes}

\end{quote}



\bibliography{scibib}

\bibliographystyle{Science}

\section*{Acknowledgments}
N/A

\section*{Supplementary materials}
N/A

\begin{figure}
\includegraphics[width=\textwidth]{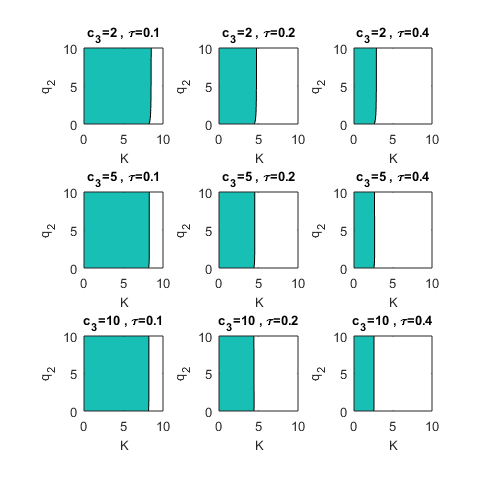}
\caption{Stability of the $5\times5$ system in the $K-q_{2}$ plane for
different values of the time scales $c_{3}$ and $\tau$. Increasing $c_{3}$ and
decreasing $\tau_{0}$ increases the region of linear stability for the
equations.}
\label{ptb}
\end{figure}

\begin{figure}
\includegraphics[width=\textwidth]{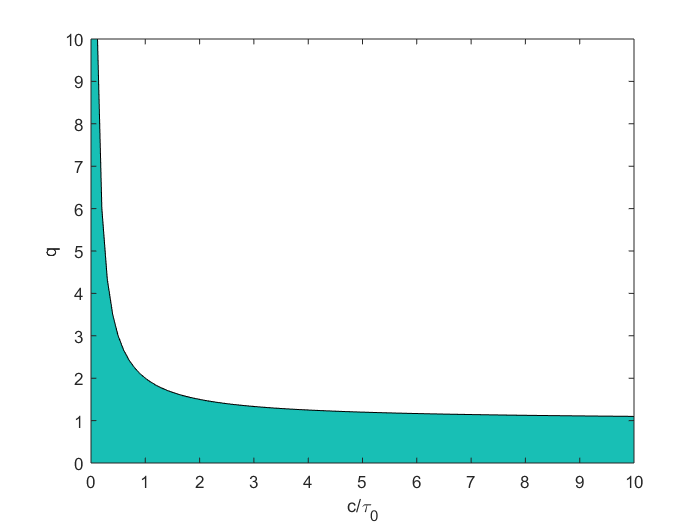}
\caption{Stability for our simplified model without the presence of
fundamental value or sentiment. The system is stable in the shaded region for
the parameters $q$ and $\frac{c}{\tau_{0}}$.}%
\label{simplestmodel}
\label{Fig 5x5}
\end{figure}

\end{document}